%% file: thesis.tex
\documentclass[11pt,letterpaper]{article}
\usepackage[lmargin=1.0in,rmargin=1.0in,bottom=1.0in,top=1.0in,twoside=False]{geometry}

\usepackage{fullpage,amssymb,amsmath}
\usepackage{graphicx}
\usepackage{enumerate}
\usepackage{tikz}
\usetikzlibrary{shapes}

\usepackage[T1]{fontenc}

 \usepackage{xcolor}
 \usepackage{mathtools}
\usepackage{microtype}
\usepackage{amsfonts}
\usepackage{comment}
\usepackage[english]{babel}
\usepackage{mathrsfs}
\usepackage[ruled,vlined]{algorithm2e}

\usepackage{trimspaces}
\usepackage{nccfoots}
\usepackage{setspace}
\usepackage{inconsolata}
\usepackage{libertine}
\usepackage[absolute]{textpos}

\usepackage{enumitem}
\usepackage{todonotes}

\usepackage{longtable}

\definecolor{blue}{rgb}{0.1,0.2,0.5}
\definecolor{brown}{rgb}{0.6,0.6,0.2}
\usepackage[ocgcolorlinks, linkcolor={blue}, citecolor={brown}]{hyperref}

\usepackage{comment}

\usepackage[amsmath,thmmarks,hyperref]{ntheorem}
\usepackage{cleveref}

\crefformat{page}{#2page~#1#3}%
\Crefformat{page}{#2Page~#1#3}%
\crefformat{equation}{#2(#1)#3}%
\Crefformat{equation}{#2(#1)#3}%
\crefformat{figure}{#2Figure~#1#3}%
\Crefformat{figure}{#2Figure~#1#3}%
\crefformat{section}{#2Section~#1#3}
\Crefformat{section}{#2Section~#1#3}
\crefformat{chapter}{#2Chapter~#1#3}
\Crefformat{chapter}{#2Chapter~#1#3}
\crefformat{chapter*}{#2Chapter~#1#3}
\Crefformat{chapter*}{#2Chapter~#1#3}
\crefformat{part}{#2Part~#1#3}
\Crefformat{part}{#2Part~#1#3}
\crefformat{enumi}{#2(#1)#3}
\Crefformat{enumi}{#2(#1)#3}

\usepackage{enumerate}

\usepackage{latexsym}

\theoremnumbering{arabic}
\theoremstyle{plain}
\theoremsymbol{}
\theorembodyfont{\itshape}
\theoremheaderfont{\normalfont\bfseries}
\theoremseparator{.}

\newtheorem{theorem}{Theorem}
\crefformat{theorem}{#2Theorem~#1#3}
\Crefformat{theorem}{#2Theorem~#1#3}

\newcommand{\newtheoremwithcrefformat}[2]{%
  \newtheorem{#1}[theorem]{#2}%
  \crefformat{#1}{##2\MakeUppercase#1~##1##3}%
  \Crefformat{#1}{##2\MakeUppercase#1~##1##3}%
}
\newcommand{\newseptheoremwithcrefformat}[2]{%
  \newtheorem{#1}{#2}%
  \crefformat{#1}{##2\MakeUppercase#1~##1##3}%
  \Crefformat{#1}{##2\MakeUppercase#1~##1##3}%
}

\newtheoremwithcrefformat{lemma}{Lemma}
\newtheoremwithcrefformat{proposition}{Proposition}
\newtheoremwithcrefformat{observation}{Observation}
\newtheoremwithcrefformat{conjecture}{Conjecture}
\newtheoremwithcrefformat{corollary}{Corollary}
\newseptheoremwithcrefformat{claim}{Claim}
\theorembodyfont{\upshape}
\newtheoremwithcrefformat{example}{Example}
\newtheoremwithcrefformat{remark}{Remark}
\newtheoremwithcrefformat{definition}{Definition}

\theoremstyle{nonumberplain}
\theoremheaderfont{\scshape}
\theorembodyfont{\normalfont}
\theoremsymbol{\ensuremath{\square}}
\newtheorem{proof}{Proof}

\theoremsymbol{\ensuremath{\lrcorner}}

\def\cqedsymbol{\ifmmode$\lrcorner$\else{\unskip\nobreak\hfil
\penalty50\hskip1em\null\nobreak\hfil$\lrcorner$
\parfillskip=0pt\finalhyphendemerits=0\endgraf}\fi}

\newcommand{\N}{\mathbb{N}}

\newcommand{\R}{\mathbb{R}}
\newcommand{\Sc}{\mathcal{S}}
\newcommand{\Uc}{\mathcal{U}}

\newcommand{\Cc}{\mathcal{C}}
\newcommand{\Dd}{\mathcal{D}}
\newcommand{\Lc}{\mathcal{L}}
\newcommand{\Ac}{\mathcal{A}}

\newcommand{\Fc}{\mathcal{F}}

\newcommand{\Ab}{\mathbb{A}}
\newcommand{\Bb}{\mathbb{B}}

\newcommand{\Oh}{\mathcal{O}}
\newcommand{\eps}{\varepsilon}

\newcommand{\tup}[1]{\bar{#1}}

\newcommand{\xb}{\bar{x}}
\newcommand{\yb}{\bar{y}}
\newcommand{\dom}{\text{dom}}
\newcommand{\Isf}{\mathsf{I}}
\newcommand{\Jsf}{\mathsf{J}}
\newcommand{\Ksf}{\mathsf{K}}

\newcommand{\FO}{\mathsf{FO}}
\newcommand{\MSO}{\mathsf{MSO}}
\newcommand{\CMSO}{\mathsf{CMSO}}
\newcommand{\CtwoMSO}{\mathsf{C}_2\mathsf{MSO}}

\newcommand{\topleft}{c_{\text{tl}}}
\newcommand{\bottomleft}{c_{\text{bl}}}

\newcommand{\floor}[1]{\left \lfloor #1 \right \rfloor }
\newcommand{\set}[1]{\left \{ #1 \right \} }

\let\originalleft\left
\let\originalright\right
\renewcommand{\left}{\mathopen{}\mathclose\bgroup\originalleft}
\renewcommand{\right}{\aftergroup\egroup\originalright}

\renewcommand{\leq}{\leqslant}

\renewcommand{\preceq}{\preccurlyeq}

\begin{document}

\title{VC density of set systems definable in tree-like graphs\thanks{An exposition of the results presented in this work can be also found in the master thesis of the first author~\cite{Paszke-thesis19}. 
    The work of Micha\l{} Pilipczuk on this article is a part of project TOTAL
    that has received funding from the European Research Council
    (ERC) under the European Union's Horizon 2020 research and
    innovation programme (grant agreement No.~677651).}}

\author{
Adam~Paszke\thanks{
  Faculty of Mathematics, Informatics, and Mechanics, University of Warsaw, Poland, \texttt{adam.paszke@gmail.com}.
}
\and
Micha\l{}~Pilipczuk\thanks{
  Institute of Informatics, University of Warsaw, Poland, \texttt{michal.pilipczuk@mimuw.edu.pl}.
}
}

\begin{titlepage}
\def\thepage{}
\thispagestyle{empty}
\maketitle

\begin{textblock}{20}(0, 12.5)
\includegraphics[width=40px]{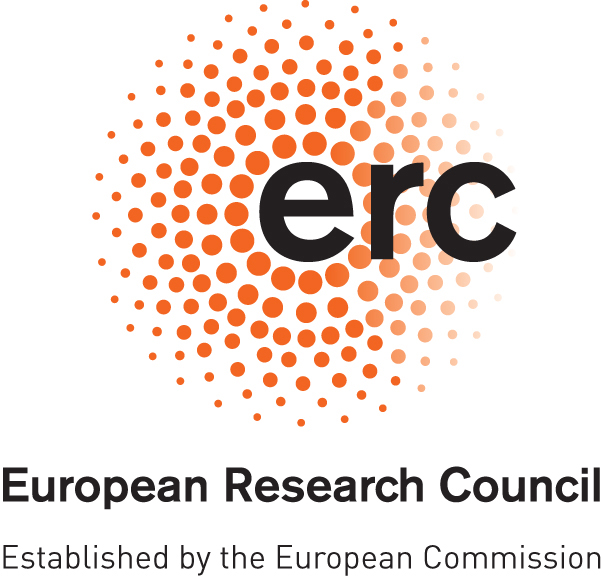}%
\end{textblock}
\begin{textblock}{20}(-0.25, 12.9)
\includegraphics[width=60px]{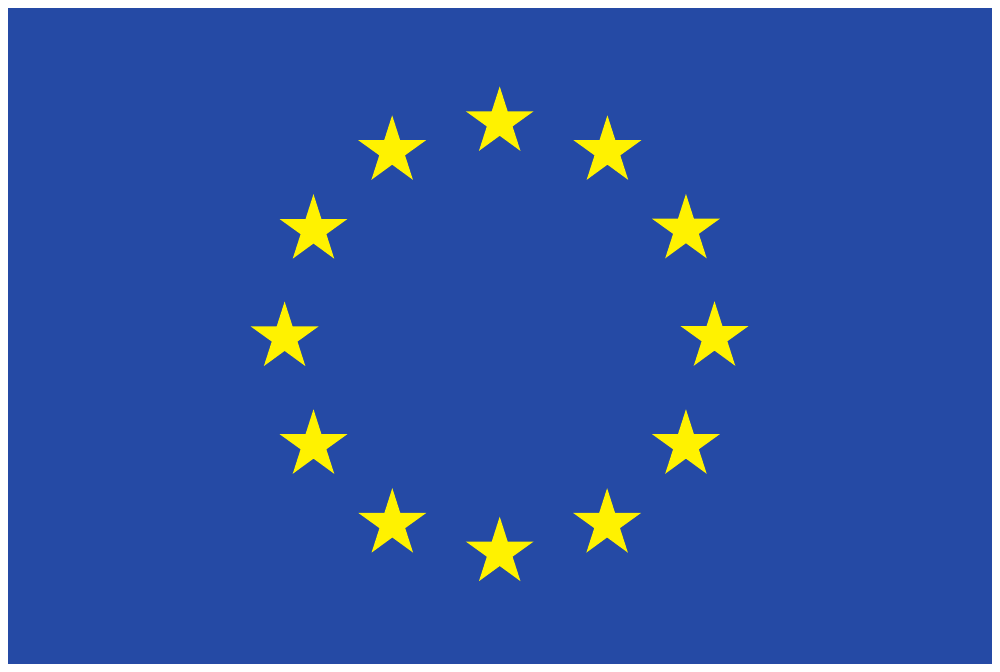}%
\end{textblock}

\input{abstract}

\end{titlepage}

\input{intro}

\input{prelims}

\input{upper}

\input{lower}

\bibliographystyle{abbrv}
\bibliography{thesis}

\end{document}

%% file: abstract.tex
\begin{abstract}
 We study set systems definable in graphs using variants of logic with different expressive power.
 Our focus is on the notion of {\em{Vapnik-Chervonenkis density}}: the smallest possible degree of a polynomial bounding the cardinalities of restrictions of such set systems.
 On one hand, we prove that if $\varphi(\tup x,\tup y)$ is a fixed $\CMSO_1$ formula and $\Cc$ is a class of graphs with uniformly bounded cliquewidth, 
 then the set systems defined by $\varphi$ in graphs from $\Cc$ have VC density at most $|\tup y|$, which is the smallest bound that one could expect.
 We also show an analogous statement for the case when $\varphi(\tup x,\tup y)$ is a $\CMSO_2$ formula and $\Cc$ is a class of graphs with uniformly bounded treewidth.
 We complement these results by showing that if $\Cc$ has unbounded cliquewidth (respectively, treewidth), then, under some mild technical assumptions on $\Cc$,
 the set systems definable by $\CMSO_1$ (respectively, $\CMSO_2$) formulas in graphs from $\Cc$ may have unbounded VC dimension, hence also unbounded VC density.
\end{abstract}

%% file: intro.tex
\section{Introduction}

\paragraph*{VC dimension.} VC dimension is a widely used parameter measuring the complexity of set systems.
Since its introduction in the 70s in the seminal work of Vapnik and Chervonenkis~\cite{Vapnik71-vc-dim}, it became a fundamental notion in statistical learning theory.
VC dimension has also found multiple applications in combinatorics and in algorithm design, particularly in the area of approximation algorithms.

The original definition states that the VC dimension of a set system $\Fc = (\Uc, \Sc)$, where $\Uc$ is the universe and $\Sc$ is the family of sets, 
is equal to the supremum of cardinalities of subsets of $\Uc$ that are shattered by $\Fc$. Here, a subset $X\subseteq \Uc$ is {\em{shattered}} by $\Fc$ 
if the {\em{restriction}} of $\Fc$ to $X$ --- defined as the set system $\Fc[X]=(X,\{S\cap X\colon S\in \Sc\})$ --- is the whole powerset of $X$.

In many applications, the boundedness of the VC dimension is exploited mainly through the {\em{Sauer-Shelah Lemma}}~\cite{Sauer72-sauer-shelah-lemma,Shelah72-sauer-shelah-lemma}, 
which states that a set system $\Fc$ over a universe of size $n$ and
of VC dimension $d$ contains only $\Oh(n^d)$ different sets. As a bound on VC dimension is inherited under restrictions, this implies that for every subset $A$ of the universe, the cardinality of the 
set system $\Fc[A]$ is at most $\Oh(|A|^d)$. This polynomial bound on the sizes of restrictions distinguishes set systems with bounded VC dimension from arbitrary set systems,
where the exponential growth is witnessed by larger and larger shattered sets.

However, for many set systems appearing in various settings, the bound provided by the Sauer-Shelah Lemma is far from optimum: the degree of the best possible polynomial bound
is much lower than the VC dimension. This motivates introducing a more refined notion of the {\em{VC density}} of a set system, 
which is (slightly informally) defined as the lowest possible degree of a polynomial bounding the cardinalities of its restrictions. See Section~\ref{sec:VC} for a formal definition.
The Sauer-Shelah Lemma then implies that the VC density is never larger than the VC dimension, but in fact it can be much lower.
This distinction is particularly important for applications in approximation algorithms, where having VC density equal to one (which corresponds to a linear bound in the Sauer-Shelah Lemma)
implies the existence of $\eps$-nets of size $\Oh(\frac{1}{\eps})$~\cite{Chan12-wsc-approx}, 
while a super-linear bound implied by the boundedness of the VC dimension gives only $\eps$-nets of size $\Oh(\frac{1}{\eps}\log \frac{1}{\eps})$ (see e.g.~\cite{Mustafa17-eps-nets}).
This difference seems innocent at first glance, but shaving off the logarithmic factor actually corresponds to the possibility of designing constant-factor approximation algorithms \cite{Chan12-wsc-approx}.
 
\paragraph*{Defining set systems in logic.}
In this work we study set systems definable in different variants of logic over various classes of graphs.
We concentrate on finding a precise understanding of the connection between the expressive power of the considered logic $\Lc$ and the structural properties of the investigated class of graphs~$\Cc$ that
are necessary and sufficient for the following assertion to hold: $\Lc$-formulas can define only simple set systems in graphs from $\Cc$, where simplicity is measured in terms of the VC parameters.

To make this idea precise, we need a way to define a set system from a graph using a formula.
Let $\varphi(\tup x,\tup y)$ be a formula of some logic $\Lc$ (to be made precise later) in the vocabulary of graphs, where $\tup x,\tup y$ are tuples of free vertex variables. 
Note here that the partition of free variables into $\tup x$ and $\tup y$ is fixed; in this case we say that $\varphi(\tup x,\tup y)$ is a {\em{partitioned formula}}.
Then $\varphi$ defines in a graph $G=(V,E)$ the set system of {\em{$\varphi$-definable sets}}:
$$S^{\varphi}(G) = \left(\ V^{\tup x}\ ,\ \{\{\tup u\in V^{\tup x}\colon G\models \varphi(\tup u,\tup v)\}\colon \tup v\in V^{\tup y}\}\ \right).$$
Here, $V^{\tup x}$ and $V^{\tup y}$ denote the sets of evaluations of variables of $\tup x$ and $\tup y$ in $V$, respectively.
In other words, every $\tup v\in V^{\tup y}$ defines the set consisting of all those $\tup u\in V^{\tup x}$ for which $\varphi(\tup u,\tup v)$ is true in $G$. 
Then $S^{\varphi}(G)$ is a set system over universe $V^{\tup x}$ that comprises all subsets of $V^{\tup x}$ definable in this way.

For an example, if $|\tup x|=|\tup y|=1$ and $\varphi(x,y)$ verifies whether the distance between $x$ and $y$ is at most $d$, for some $d\in \N$, then $S^{\varphi}(G)$ is the set system whose universe is the vertex set of $G$,
while the set family comprises all balls of radius $d$ in $G$.

The situation when the considered logic $\Lc$ is the First Order logic $\FO$ was recently studied by Pilipczuk, Siebertz, and Toru\'nczyk~\cite{Pilipczuk18-types-in-sparse-graphs}.
They showed that the simplicity of $\FO$-definable set systems in graphs is tightly connected to their sparseness, as explained formally next.
On one hand, if $\Cc$ is a {\em{nowhere dense}}\footnote{{\em{Nowhere denseness}} is a notion of uniform sparseness in graphs. As it is not directly related to our investigations, we refrain from
giving a formal definition, and refer the interested reader to the discussion in~\cite{Pilipczuk18-types-in-sparse-graphs} instead.}
class of graphs, 
then for every partitioned $\FO$ formula $\varphi(\tup x,\tup y)$, $\varphi$ defines in graphs from $\Cc$ set systems of VC density at most $|\tup y|$.
On the other hand, if $\Cc$ is not nowhere dense, but is closed under taking subgraphs, then there exists a partitioned $\FO$ formula that defines in graphs from $\Cc$ set systems of arbitrarily high
VC dimension, hence also arbitrarily high VC density.
Note that one cannot expect lower VC density than $|\tup y|$ for any non-trivial logic $\Lc$ and class $\Cc$, because already the very simple formula $\alpha(x,\tup y)=\bigvee_{i=1}^{|\tup y|}\,(x=y_i)$ 
defines set systems of VC density $|\tup y|$ in edgeless graphs. Thus, in some sense the result stated above provides a sharp dichotomy.

In this work we are interested in similar dichotomy statements for more expressive variants of logic on graphs, namely $\MSO_1$ and $\MSO_2$.
Recall that $\MSO_1$ on graphs extends $\FO$ by allowing quantification over subsets of vertices, while in $\MSO_2$ one can in addition quantify over subsets of edges.
This setting has been investigated by Grohe and Tur\'an~\cite{Grohe04-mso-tree-vc-dim}. 
They proved that if graphs from a graph class $\Cc$ have uniformly bounded cliquewidth (i.e. there is a constant $c$ that is an upper bound on the cliquewidth of every member of $\Cc$), then every $\MSO_1$ formula
defines in graphs from $\Cc$ set systems with uniformly bounded VC dimension. They also gave a somewhat complementary lower bound showing that if $\Cc$ contains graphs of arbitrarily high treewidth and is closed under
taking subgraphs, then there exists a fixed $\MSO_1$ formula that defines in graphs from $\Cc$ set systems with unbounded VC dimension.

\paragraph*{Our contribution.} We improve the results of Grohe and Tur\'an~\cite{Grohe04-mso-tree-vc-dim} in two aspects. 
First, we prove tight upper bounds on the VC density of the considered set systems, and not only on the VC dimension.
Second, we clarify the dichotomy statements by showing that the boundedness of the VC parameters for set systems definable in $\MSO_1$ is tightly connected to the boundedness of cliquewidth,
and there is a similar connection between the complexity of set systems definable in $\MSO_2$ and the boundedness of treewidth.
Formal statements follow.

For the upper bounds, our results are captured by the following theorem. Here, $\CMSO_1$ and $\CMSO_2$ are extensions of $\MSO_1$ and $\MSO_2$, respectively, by {\em{modular predicates}} of the form 
$|X|\equiv a\bmod p$, where $X$ is a monadic variable and $a,p$ are integers. Also, $\CtwoMSO_1$ is a restriction of $\CMSO_1$ where we allow only modular predicates with $p=2$, that is, checking the parity of
the cardinality of a set.

\begin{theorem}\label{thm:upper-bounds}
    Let $\Cc$ be a class of graphs and $\varphi(\xb, \yb)$ be a partitioned formula. Additionally, assume that one of the following assertions holds:
    \begin{itemize}[nosep]
        \item[(i)] $\Cc$ has uniformly bounded cliquewidth and $\varphi(\xb, \yb)$ is a $\CMSO_1$-formula; or
        \item[(ii)] $\Cc$ has uniformly bounded treewidth and $\varphi(\xb, \yb)$ is a $\CMSO_2$-formula.
    \end{itemize}
    Then there is a constant $c\in \N$ such that for every graph $G\in \Cc$ and non-empty vertex subset $A\subseteq V(G)$,
    $$|S^{\varphi}(G)[A]|\leq c\cdot |A|^{|\tup y|}.$$
\end{theorem}

In particular, this implies that for a partitioned formula $\varphi(\tup x,\tup y)$, 
the class of set systems $S^{\varphi}(\Cc)$ has VC density $|\tup y|$ whenever $\Cc$ has uniformly bounded cliquewidth and
$\varphi$ is a $\CMSO_1$-formula, or $\Cc$ has uniformly bounded treewidth and $\varphi$ is a $\CMSO_2$-formula. 

Note that \cref{thm:upper-bounds} provides much better bounds on the cardinalities of restrictions of the considered set systems than bounding the VC dimension and using the Sauer-Shelah Lemma, 
as was done in~\cite{Grohe04-mso-tree-vc-dim}.
In fact, as argued in~\cite[Theorem 12]{Grohe04-mso-tree-vc-dim}, even in the case of defining set systems over words, the VC dimension can be tower-exponential high with respect to the size of the formula.
In contrast, \cref{thm:upper-bounds} implies that the VC density will be actually much lower: at most $|\tup y|$.
This improvement has an impact on some asymptotic bounds in learning-theoretical corollaries discussed by Grohe and Tur\'an, see e.g.~\cite[Theorem~1]{Grohe04-mso-tree-vc-dim}.

For lower bounds, we work with labelled graphs. For a finite label set $\Lambda$, a {\em{$\Lambda$-v-labelled graph}} is a graph whose vertices are labelled using labels from $\Lambda$,
while in a {\em{$\Lambda$-ve-labelled graph}} we label both the vertices and the edges using $\Lambda$.
For a graph class $\Cc$, by $\Cc^{\Lambda,1}$ we denote the class of all $\Lambda$-v-labelled graphs whose underlying unlabeled graphs belong to $\Cc$, while $\Cc^{\Lambda,2}$ is
defined analogously  for $\Lambda$-ve-labelled graphs. The discussed variants of $\MSO$ work over labelled graphs in the obvious way.

\begin{theorem}
    \label{thm:lower-bounds}
    There exists a finite label set $\Lambda$ such that the following holds.
    Let $\Cc$ be a class of graphs and $\Lc$ be a logic such that either
    \begin{itemize}[nosep]
        \item[(i)] $\Cc$ contains graphs of arbitrarily large cliquewidth and $\Lc=\CtwoMSO_1$; or
        \item[(ii)] $\Cc$ contains graphs of arbitrarily large treewidth and $\Lc=\MSO_2$.
    \end{itemize}
    Then there exists a partitioned $\Lc$-formula $\varphi(x,y)$ in the vocabulary of graphs from $\Cc^{\Lambda,t}$, where $t=1$ if (i) holds and $t=2$ if (ii) holds,
    such that the family 
    $$\{\ S^\varphi(G)\,\colon\, G\in \Cc^{\Lambda,t}\ \},$$
    contains set systems with arbitrarily high VC dimension.
\end{theorem}

Thus, the combination of \cref{thm:upper-bounds} and~\cref{thm:lower-bounds} provides a tight understanding of 
the usual connections between $\MSO_1$ and cliquewidth, and between $\MSO_2$ and treewidth, also in the setting of definable set systems. 
We remark that the second connection was essentially observed by Grohe and Tur\'an in~\cite[Corollary~20]{Grohe04-mso-tree-vc-dim}, whereas the first seems new, but follows from a very similar argument.

As argued by Grohe and Tur\'an in~\cite[Example~21]{Grohe04-mso-tree-vc-dim}, some mild technical conditions, like closedness under labelings with a finite label set, is necessary for
a result like \cref{thm:lower-bounds} to hold. Indeed, the class of $1$-subdivided complete graphs has unbounded treewidth and cliquewidth, yet $\CMSO_1$- and $\CMSO_2$-formulas can only define
set systems of bounded VC dimension on this class, due to symmetry arguments. Also, the fact that in the case of unbounded cliquewidth we need to rely on logic $\CtwoMSO_1$ instead of plain $\MSO_1$ is connected
to the longstanding conjecture of Seese~\cite{Seese91-mso-tw-undecidable} about decidability of $\MSO_1$ in classes of graphs.


%% file: prelims.tex
\section{Preliminaries}\label{sec:prelims}

\subsection{Vapnik-Chervonenkis parameters}\label{sec:VC}
In this section we briefly recall the main definitions related to the Vapnik-Chervonenkis parameters.
We only provide a terse summary of the relevant concepts and results, and refer to the work of Mustafa and Varadarajan \cite{Mustafa17-eps-nets} for a broader context.

A {\em{set system}} is a pair $\Fc=(\Uc, \Sc)$, where $\Uc$ is the {\em{universe}} or {\em{ground set}}, while $\Sc$ is a family of subsets of $\Uc$.
While a set system is formally defined as the pair $(\Uc, \Sc)$, we will often use that term with a family $\Sc$ alone, and then $\Uc$ is implicitly taken to be $\bigcup_{S \in \Sc} S$.
The {\em{size}} of a set system is $|\Fc|\coloneqq |\Sc|$.

For a set system $\Fc=(\Uc,\Sc)$ and $X\subseteq \Uc$, the {\em{restriction}} of $\Sc$ to $X$ is the set system $\Fc[X]\coloneqq (X,\Sc\cap X)$, where $\Sc\cap X\coloneqq \{S\cap X\colon S\in \Sc\}$.
We say that $X$ is \textit{shattered} by $\Fc$ if $\Sc\cap X$ is the whole powerset of $X$.
Then the {\em{VC dimension}} of $\Fc$ is the supremum of cardinalities of sets shattered by $\Fc$.

As we are mostly concerned with the asymptotic behavior of restrictions of set systems, the following notion will be useful.
\begin{definition}\label{def:shatter-function}
    The {\em{growth function}} of a set system $\Fc=(\Uc,\Sc)$ is the function $\pi_{\Fc}\colon \N\to \N$ defined as:
    $$\pi_{\Fc}(n) = \max\ \{\ |\Sc \cap X|\, \colon\, X \subseteq \Uc,\, |X| = n\ \}\qquad \textrm{for }n\in \N.$$
\end{definition}

Clearly, for any set system $\Fc$ we have that $\pi_{\Fc}(n) \leq 2^n$, but many interesting set systems admit asymptotically polynomial bounds.
This is in particular implied by the boundedness of the VC dimension, via the Sauer-Shelah Lemma stated below.

\begin{lemma}[Sauer--Shelah Lemma~\cite{Sauer72-sauer-shelah-lemma,Shelah72-sauer-shelah-lemma}]\label{lem:sauer-shelah}
    If $\Fc$ is a set system of VC dimension $d$, then
    $$\pi_{\Fc}(n) \leq \binom{n}{0}+\binom{n}{1}+\ldots+\binom{n}{d} \leq \Oh(n^d).$$
\end{lemma}

Note that when the VC dimension of $\Fc$ is not bounded, then for every $n$ there is a set of size $n$ that is shattered by $\Fc$, which implies that $\pi_{\Fc}(n)=2^n$.
This provides an interesting dichotomy: if $\pi_{\Fc}(n)$ is not bounded by a polynomial, it must be equal to the function $2^n$.

As useful as the Sauer--Shelah Lemma is, the upper bound on asymptotics of the growth function implied by it is quite weak for many natural set systems.
Therefore, we will study the following quantity.

\begin{definition}
    The {\em{VC density}} of a set system $\Fc$ is the quantity
    $$\inf\,\{\ \alpha\in \R^+\ \colon\ \textrm{there exists }c\in \R\textrm{ such that }\pi_{\Fc}(n)\leq c\cdot n^\alpha\textrm{ for all }n\in \N\ \}.$$
\end{definition}

Observe that the definition of the VC density of $\Fc$ makes little sense when the universe of $\Fc$ is finite, as then the growth function ultimately becomes $0$, allowing a polynomial bound of arbitrary small degree.
Therefore, we extend the definition of VC density to {\em{classes}} of finite set systems (i.e., families of finite set systems) as follows:
the VC density of a class $\Cc$ is the infimum over all $\alpha\in \R^+$ for which there is $c\in \R$ such that $\pi_{\Fc}(n)\leq c\cdot n^\alpha$ for all $\Fc\in \Cc$ and $n\in \N$.
Note that this is equivalent to measuring the VC density of the set system obtained by taking the union of all set systems from $\Cc$ on disjoint universes.
Similarly, the VC dimension of a class of set systems $\Cc$ is the supremum of the VC dimensions of the members of $\Cc$.

Thus, informally speaking the VC density of $\Fc$ is the lowest possible degree of a polynomial bound that fits the conclusion of the Sauer--Shelah lemma for $\Fc$.
Clearly, the Sauer--Shelah lemma implies that the VC density is never larger than the VC dimension, but as it turns out, that connection goes both ways:

\begin{lemma}[\cite{Mustafa17-eps-nets}]
    A set system $\Fc$ satisfying $\pi_{\Fc}(n) \leq cn^d$ for all $n\in \N$ has VC dimension bounded by $4d\log(cd)$.
\end{lemma}

Hence, a set system $\Fc$ has finite VC dimension if and only if it has finite VC density, but the results showing their equivalence usually produce relatively weak bounds.
As discussed in the introduction, VC density is often a finer measure of complexity than VC dimension for interesting problems.


\subsection{Set systems definable in logic}

We assume basic familiarity with relational structures.
The {\em{domain}} (or {\em{universe}}) of a relational structure $\Ab$ will be denoted by $\dom(\Ab)$.
For a tuple of variables $\tup x$ and a subset $S\subseteq \dom(\Ab)$, by $S^{\tup x}$ we denote the set of all {\em{evaluations}} of $\tup x$ in $S$, that is, functions mapping the variables of $\tup x$ to elements of $S$.
A {\em{class}} of structures is a set of relational structures over the same signature.

Consider a logic $\Lc$ over some relational signature $\Sigma$.
A {\em{partitioned formula}} is an $\Lc$-formula of the form $\varphi(\xb, \yb)$, where the free variables are partitioned into \textit{object variables} $\xb$ and \textit{parameter variables} $\yb$.
Then for a $\Sigma$-structure $\Ab$, we can define the set system of {\em{$\varphi$-definable sets}} in $\Ab$:
$$S^{\varphi}(\Ab) = \left(\ \dom(\Ab)^{\tup x}\ ,\ \{\{\tup u\in \dom(\Ab)^{\tup x}\colon \Ab\models \varphi(\tup u,\tup v)\}\colon \tup v\in \dom(\Ab)^{\tup y}\}\ \right).$$
If $\Cc$ is a class of $\Sigma$-structures, then we define the class of set systems $S^{\varphi}(\Cc)\coloneqq \{S^{\varphi}(\Ab)\colon \Ab\in \Cc\}$.

Note that the universe of $S^{\varphi}(\Ab)$ is $\dom(\Ab)^{\tup x}$, so the elements of $S^{\varphi}(\Ab)$ can be interpreted as tuples of elements of $\Ab$ of length $|\tup x|$.
When measuring the VC parameters of set systems $S^{\varphi}(\Ab)$ it will be convenient to somehow still regard $\dom(\Ab)$ as the universe.
Hence, we introduce the following definition: a {\em{$k$-tuple set system}} is a pair $(\Uc,\Sc)$, where $\Uc$ is a universe and $\Sc$ is a family of sets of $k$-tuples of elements of $\Uc$.
Thus, $S^{\varphi}(\Ab)$ can be regarded as an $|\tup x|$-tuple set system with universe $\dom(\Ab)$.

When $\Fc=(\Uc,\Sc)$ is a $k$-tuple set system, for a subset of elements $X\subseteq \Uc$ we define
$$\Sc\cap X\coloneqq \{S\cap X^k\colon S\in \Sc\}.$$
This naturally gives us the definition of a restriction: $\Fc[X]\coloneqq (X,\Sc\cap X)$. We may now lift all the relevant definitions --- of shattering, of the VC dimension, of the growth function, and of the VC density ---
to $k$-tuple set systems using only such restrictions: to subsets $X\subseteq \Uc$.
Note that these notions for $k$-tuple set systems are actually different from the corresponding regular notions, which would consider $\Fc$ as a set system with universe $\Uc^k$.
This is because, for instance for the VC dimension, in the regular definition we would consider shattering all possible subsets of $k$-tuples of the universe, while in the definition for $k$-tuple set systems we restrict attention
to shattering sets of the form $X^k$, where $X\subseteq \Uc$.

\subsection{MSO and transductions}

Recall that Monadic Second Order logic ($\MSO$) is an extension of the First Order logic ($\FO$) that additionally allows quantification over subsets of the domain (i.e. unary predicates), represented as {\em{monadic variables}}.
Sometimes we will also allow {\em{modular predicates}} of the form $|X|\equiv a \bmod p$, where $X$ is a monadic variable and $a,p$ are integers,
in which case the corresponding logic shall be named $\CMSO$.
If only parity predicates may be used (i.e. $p=2$), we will speak about $\CtwoMSO$ logic.

The main idea behind the proofs presented in the next sections is that we will analyze how complicated set systems one can define in $\MSO$ on specific simple structures: trees and grid graphs.
Then these results will be lifted to more general classes of graphs by means of \textit{logical transductions}.

For a logic $\Lc$ (usually a variant of $\MSO$) and a signature $\Sigma$, by $\Lc[\Sigma]$ we denote the logic comprising all $\Lc$-formulas over $\Sigma$. Then deterministic $\Lc$-transductions are defined as follows.

\begin{definition}
    Fix two relational signatures $\Sigma$ and $\Sigma' = (R_1, \ldots, R_k)$.
    A \textit{deterministic $\Lc$-transduction} $\Isf$ from $\Sigma$-structures to $\Sigma'$-structures is a sequence of $\Lc[\Sigma]$-formulas: $\gamma(x), \theta_{R_1}(\xb_1), \ldots, \theta_{R_k}(\xb_k)$, where
    the length of $\xb_i$ matches the arity of $R_i$.
\end{definition}

The semantics we associate with this definition is as follows. 
Let $\Ab$ be a $\Sigma$ structure and $D = \set{ u \colon u \in \dom(\Bb), \Bb \models \gamma(u) }$.
Then $\Isf(\Ab)$ is a $\Sigma'$ structure given by:
$$\left\langle\  D,\  \set{ \tup u_1\, \colon\, \tup u_1 \in D^{\tup x_i}, \Ab \models \theta_{R_1}(\tup x_1) },\ \ldots,\ \set{ \tup u_k\, \colon\, \tup u_k \in D^{\tup x_k}, \Ab \models \theta_{R_k}(\tup x_k) }\ \right\rangle.$$
In a nutshell, we restrict the universe of the input structure to the elements satisfying $\gamma(x)$, 
and in this new domain we reinterpret the relations of $\Sigma'$ using $\Lc[\Sigma]$-formulas evaluated in $\Ab$.

We will sometimes work with {\em{non-deterministic transductions}}, which are the following generalization.

\begin{definition}
    Fix two relational signatures $\Sigma$ and $\Sigma'$.
    A \textit{non-deterministic $\Lc$-transduction} $\Isf$ from $\Sigma$-structures to $\Sigma'$-structures is a pair consisting of: a finite signature $\Gamma(\Isf)$ consisting entirely of unary relation symbols, which
    is disjoint from $\Sigma\cup \Sigma'$; and a deterministic $\Lc$-transduction $\Isf'$ from $\Sigma\cup \Gamma(\Isf)$-structures to $\Sigma'$-structures. Transduction $\Isf'$ is called the {\em{deterministic part}} of $\Isf$.
\end{definition}

We associate the following semantics with this definition. If $\Ab$ is a $\Sigma$-structure, then by $\Ab^{\Gamma(\Isf)}$ 
we denote the set of all possible $\Sigma\cup \Gamma(\Isf)$-structures obtained by adding valuations of the unary predicates from $\Gamma(\Isf)$ to $\Ab$.
Then we define $\Isf(\Ab)\coloneqq \Isf'(\Ab^{\Gamma(\Isf)})$, which is again a set of structures.
Thus, a non-deterministic transduction $\Isf$ can be seen as a procedure that first non-deterministically selects the valuation of the unary predicates from $\Gamma(\Isf)$ in the input structure, 
and then applies the deterministic part.

If $\Cc$ is a class of $\Sigma$-structures and $\Isf$ is a transduction (deterministic or not), then by $\Isf(\Cc)$ we denote the sum of images of $\Isf$ over elements of $\Cc$.
Also, if $\Gamma$ is a signature consisting of unary relation names that is disjoint from $\Sigma$, 
then we write $\Cc^{\Gamma}\coloneqq \{\Ab^{\Gamma}\colon \Ab\in \Cc\}$ for the class of all possible $\Sigma\cup \Gamma$-structures that can be obtained from the structures from $\Cc$ by adding valuations of 
the unary predicates from $\Gamma$.

An important property of deterministic transductions is that $\MSO$ formulas working over the output structure can be ``pulled back'' to $\MSO$ formulas working over the input structure that select exactly the same tuples.
All one needs to do is add guards for all variables, ensuring that the only entities we operate on are those accepted by $\gamma(x)$, 
and replace all relational symbols of $\Sigma'$ with their respective formulas which define the transduction. This translation is formally encapsulated in the following result.

\begin{lemma}[Backwards Translation Lemma, cf.~\cite{Courcelle94-mso-transduction-survey}]\label{lem:btl}
    Let $\Isf$ be a deterministic transduction from $\Sigma$-structures to $\Sigma'$-structures, and let $\Lc\in \{\MSO,\CMSO,\CtwoMSO\}$.
    Then for every $\Lc[\Sigma']$-formula $\varphi(\tup x)$ there is an $\Lc[\Sigma]$-formula $\psi(\tup x)$ such that for every $\Sigma$-structure $\Ab$ and $\tup u\in \dom(\Ab)^{\tup x}$,
    $$\Ab\models \psi(\tup u)\qquad\qquad \textrm{if and only if}\qquad\qquad \tup u\in \dom(\Isf(\Ab))^{\tup x}\ \textrm{ and }\ \Isf(\Ab)\models \varphi(\tup u).$$
\end{lemma}

The formula $\psi$ provided by Lemma~\ref{lem:btl} will be denoted by $\Isf^{-1}(\varphi)$.


Finally, we remark that in the literature there is a wide variety of different notions of logical transductions and interpretations; we chose one of the simplest, as it will be sufficient for our needs.
We refer a curious reader to a survey of Courcelle~\cite{Courcelle94-mso-transduction-survey}.

\subsection{MSO on graphs}

We will work with two variants of $\MSO$ on graphs: $\MSO_1$ and $\MSO_2$.
Both these variants are defined as the standard notion of $\MSO$ logic, but applied to two different encodings of graphs as relational structures.
When we talk about $\MSO_1$-formulas, we mean $\MSO$-formulas over structures representing graphs as follows: elements of the structure correspond to vertices and there is a single binary relation representing adjacency.
The second variant, $\MSO_2$, encompasses $\MSO$-formulas over structures representing graphs as follows: the domain contains both edges and vertices of the graph, and there is a binary incidence relation that selects all pairs $(e,u)$ such that $e$ is
an edge and $u$ is one of its endpoints.
These two encodings of graphs will be called the {\em{adjacency encoding}} and the {\em{incidence encoding}}, respectively.

Thus, practically speaking, in $\MSO_1$ we may only quantify over subsets of vertices, while in $\MSO_2$ we allow quantification both over subsets of vertices and over subsets of edges.
$\MSO_2$ is strictly more powerful than $\MSO_1$, for instance it can express that a graph is Hamiltonian.
We may extend $\MSO_1$ and $\MSO_2$ with modular predicates in the natural way, thus obtaining logic $\CMSO_1$, $\CtwoMSO_1$, etc.

If $G$ is a graph and $\varphi(\tup x,\tup y)$ is an $\Lc$-formula over graphs, where $\Lc$ is any of the variants of $\MSO$ discussed above, then we may define the $|\tup x|$-tuple set system $S^{\varphi}(G)$ as before, where
the universe of $S^{\varphi}(G)$ is the vertex set of $G$.
We remark that in case of $\MSO_2$, despite the fact that formally an $\MSO_2$-formula works over a universe consisting of both vertices and edges, 
in the definition of $S^{\varphi}(G)$ we consider only the vertex set $V$ as the universe. That is, the parameter variables $\tup y$ range over $V$ and each evaluation $\tup v\in V^{\tup y}$ defines
the set of evaluations $\tup u\in V^{\tup x}$ satisfying $G\models \varphi(\tup u,\tup v)$ which is included in $S^{\varphi}(G)$.

\subsection{MSO and tree automata}

When proving upper bounds we will use the classic connection between $\MSO$ and tree automata.
Throughout this paper, all trees will be finite, rooted, and binary: every node may have a left child and a right child, though one or both of them may be missing. 
Trees will be represented as relational structures where the domain consists of the nodes and there are two binary relations, respectively encoding being a left child and a right child.
In case of labeled trees, the signature is extended with a unary predicate for each label.

\begin{definition}
    Let $\Sigma$ be a finite alphabet. 
    A \textit{(deterministic) tree automaton} is a tuple $(Q, F, \delta)$ where $Q$ is a finite set of states, $F$ is a subset of $Q$ denoting the accepting states, 
    while $\delta \colon (Q \cup \set{\bot})^2 \times \Sigma \to Q$ is the transition function.
\end{definition}

A {\em{run}} of a tree automaton $\Ac = (Q, F, \delta)$ over a $\Sigma$-labeled tree $T$ is the labeling of its nodes $\rho \colon V(T) \to Q$ which is computed in a bottom-up manner using the transition function.
That is, if a node $v$ bears symbol $a\in \Sigma$ and the states assigned by the run to the children of $v$ are $q_1$ and $q_2$, respectively, then the state assigned to $v$ is $\delta(q_1,q_2,a)$.
In case $x$ has no left or right child, the corresponding state $q_t$ is replaced with the special symbol $\bot$.
In particular, the state in every leaf is determined as $\delta(\bot,\bot,a)$, where $a\in \Sigma$ is the label of the leaf.
We say that a tree automaton $\Ac$ {\em{accepts}} a finite tree $T$ if $\rho(\text{root}(T)) \in F$.

The following statement expresses the classic equivalence of $\CMSO$ and finite automata over trees.

\begin{lemma}[\cite{Rabin69-s2s-decidability}]
    For every $\CMSO$ sentence $\varphi$ over the signature of $\Sigma$-labeled trees there exists a tree automaton $\Ac_\varphi$ which is equivalent to $\varphi$ in the following sense:
    for every $\Sigma$-labeled tree $T$, $T \models \varphi$ if and only if $\Ac_\varphi$ accepts $T$.
\end{lemma}

Since we are actually interested in formulas with free variables and not only sentences, we will need to change this definition slightly.
Informally speaking, we will enlarge the alphabet in a way which allows us to encode valuations of the free variables.
Let $T$ be a $\Sigma$-labelled tree and consider a tuple of variables $\xb$ along with its valuation $\tup u \in V(T)^{\xb}$.
Then we can encode $\tup u$ in $T$ by defining the {\em{augmented tree}} $T_{\tup a}$ as follows: $T_{\tup a}$ is the tree with labels from $\Sigma \times \set{0, 1}^{\xb}$ that is obtained from $T$
by enriching the label of every node $v$ with the function $f_v\in \set{0,1}^{\xb}$ defined as follows: for $x\in \tup x$, we have $f_v(x)=1$ if and only if $v=\tup u(x)$.
As observed by Grohe and Tur\'an~\cite{Grohe04-mso-tree-vc-dim}, $\CMSO$ formulas can be translated to equivalent tree automata working over augmented trees.

\begin{lemma}[\cite{Grohe04-mso-tree-vc-dim}]
    \label{lem:grohe-mso-to-automaton}
    For every $\CMSO$ formula $\varphi(\xb)$ over the signature of $\Sigma$-labeled trees there exists a tree automaton $\Ac_\varphi$ over $\Sigma \times \set{0, 1}^{\xb}$-labelled trees 
    which is equivalent to $\varphi(\xb)$ in the following sense: for every $\Sigma$-labelled tree $T$ and $\tup{u}\in V(T)^{\tup x}$, $T \models \varphi(\tup u)$ if and only if $\Ac_\varphi$ accepts $T_{\tup u}$.
\end{lemma}

%% file: upper.tex
\newcommand{\bl}{\diamond}

\section{Upper bounds}

In this section we prove \cref{thm:upper-bounds}.
We start with investigating the case of $\CMSO$-definable set systems in trees. 
This case will be later translated to the case of classes with bounded treewidth or cliquewidth by means of $\CMSO$-transductions.

\subsection{Trees}

Recall that labelled binary trees are represented as structures with domains containing their nodes, two successor relations---one for the left child, and one for the right---and unary predicates for labels.
It turns out that $\CMSO$-definable set systems over labelled trees actually admit optimal upper bounds for VC density. This improves the result of Grohe and Tur\'an~\cite{Grohe04-mso-tree-vc-dim} showing that such set systems
have bounded VC dimension.

\begin{theorem}
    \label{thm:vc-density-in-trees}
    Let $\Cc$ be a class of finite binary trees with labels from a finite alphabet $\Sigma$, and $\varphi(\xb, \yb)$ be a partitioned $\CMSO$-formula over the signature of $\Sigma$-labeled binary trees. 
    Then there is a constant $c\in \N$ such that for every tree $T\in \Cc$ and a non-empty subset of its nodes $A$, we have
    $$|S^\varphi(T)[A]|\leq c\cdot |A|^{|\yb|}.$$
\end{theorem}

\begin{proof}
    By \cref{lem:grohe-mso-to-automaton}, $\varphi(\xb, \yb)$ is equivalent to a tree automaton $\Ac = (Q, F, \delta)$ over an alphabet of $\Sigma \times \set{0, 1}^{\xb} \times \set{0, 1}^{\yb}$.
    We will now investigate how the choice of parameters $\yb$ can affect the runs of $\Ac$ over $T$.

    Since we are really considering $T$ over the alphabet extended with binary markers for $\xb$ and $\yb$, we will use $T_{\bl\bl}$ to denote the extension of the labeling of $T$ where all binary markers are set to $0$.
    That is, $T_{\bl\bl}$ is the tree labeled with alphabet $\Sigma \times \set{0, 1}^{\xb} \times \set{0, 1}^{\yb}$ obtained from $T$ by extending each symbol appearing in $T$ with functions that map all variables of $\xb$
    and $\yb$ to $0$. Tree $T_{\bl\bar{q}}$ is defined analogously, where the markers for $\yb$ are set according to the valuation $\tup q$, while the markers for $\xb$ are all set to $0$.

    In $T$ we have natural ancestor and descendant relations; we consider every node its own ancestor and descendant as well.
    Let $B$ be the subset of nodes of $T$ that consists of:
    \begin{itemize}[nosep]
     \item the root of $T$;
     \item all nodes of $A$; and
     \item all nodes $u\notin A$ such that both the left child and right child of $u$ have a descendant that belongs to $A$.
    \end{itemize}
    Note that $|B| \leq 1+|A| + (|A| - 1) = 2 |A|$. For convenience, let $\phi\colon V(T)\to B$ be a function that maps every node $u$ of $T$ to the least ancestor of $u$ that belongs to $B$.
    
    We define a tree $T'$ with $B$ as the set of nodes as follows.
    A node $v\in B$ is the left child of a node $u\in B$ in $T'$
    if the following holds in $T$: $v$ is a descendant of the left child of $u$ and no internal vertex on the unique path from $u$ to $v$ belongs to $B$.
    Note that every node $u\in B$ has at most one left child in $T'$, for if it had two left children $v,v'$, then the least common ancestor of $v$ and $v'$ would belong to $B$ and would be an internal vertex on
    both the $u$-to-$v$ path and the $u$-to-$v'$ path.
    The right child relation in $T'$ is defined analogously.
    The reader may think of $T'$ as of $T$ with $\phi^{-1}(u)$ contracted to $u$, for every $u\in B$; see \cref{fig:contract}.
    
    \begin{figure}[h]
        \centering
        \includegraphics[height=4cm]{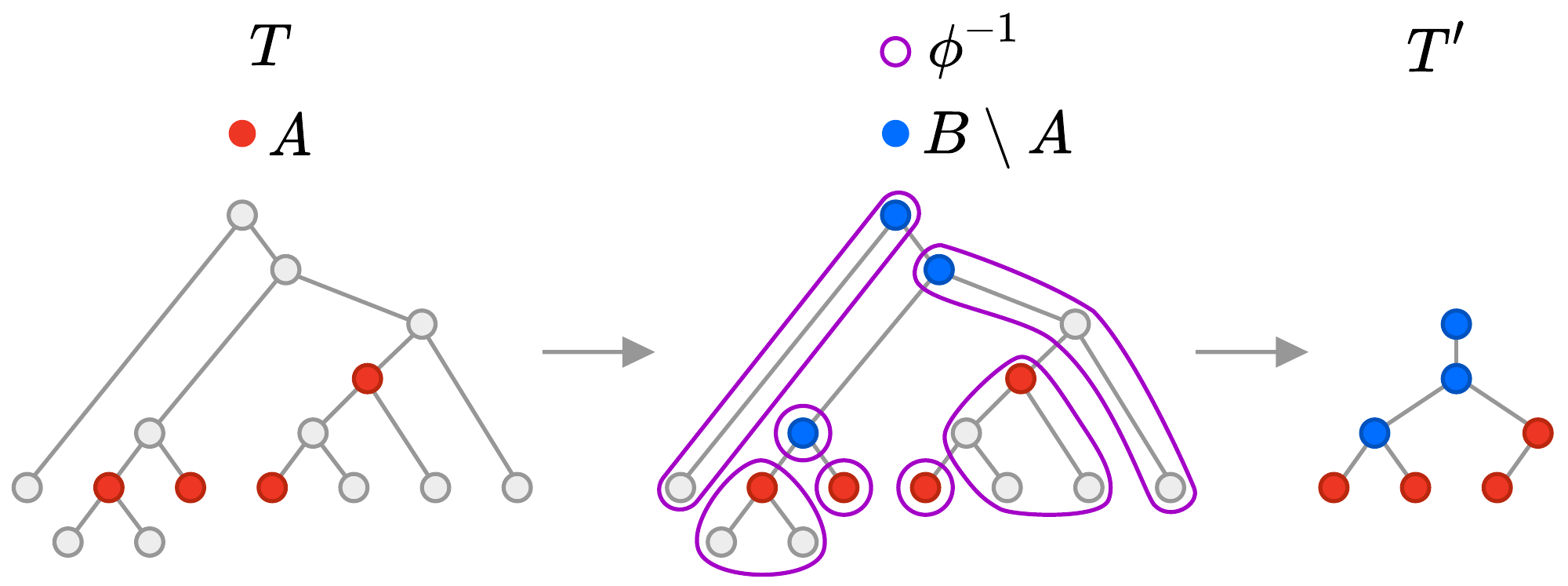}
        \caption{Definitions of $B$, $\phi$, and $T'$.}\label{fig:contract}
    \end{figure}

    Note that we did not define any labeling on the tree $T'$. Indeed, we treat $T'$ as an unlabeled tree, but will consider different labelings of $T'$ induced by various augmentations of $T$.
    For this, we define alphabet
    $$\Delta = \set{0, 1}^{\xb} \to \left((Q^2 \to Q) \cup (Q \to Q) \cup Q\right),$$
    where $X\to Y$ denotes the set of functions from $X$ to $Y$.     
    Now, for a fixed valuation of parameter variables $\tup q\in V(T)^{\yb}$ and object variables $\tup p\in V(T)^{\xb}$, we define the $\Delta$-labeled tree $T'_{\tup q}$ as follows.
    Consider any node $u\in B$ and let $T_{\tup p\tup q}[u]$ be the {\em{context}} of $u$: a tree obtained from $T_{\tup p\tup q}$ by restricting it to the descendants of $u$, and, for every child $v$ of $u$ in~$T'$,
    replacing the subtree rooted at $v$ by a single special node called a {\em{hole}}. 
    The automaton $\Ac$ can be now run on the context $T_{\tup p\tup q}[u]$ provided that for every hole of $T_{\tup p\tup q}[u]$ we prescribe a state to which this hole should evaluate.
    Thus, running $\Ac$ on $T_{\tup p\tup q}[u]$ defines a state transformation $\delta'_{\tup p\tup q}[u]$, 
    which maps tuples of states assigned to the holes of $T_{\tup p\tup q}[u]$ to the state assigned to $u$.
    Intuitively, $\delta'_{\tup p\tup q}[u]$ encodes the compressed transition function of $\Ac$ when run over the subtree of $T_{\tup p\tup q}$ induced by $\phi^{-1}(u)$, where it is assumed that on the input we are given 
    the states to which the children of $u$ in $T'$ are evaluated. 
    Note that the domain of $\delta'_{\tup p\tup q}[u]$ consists of pairs of states if $u$ has two children in $T'$, of one state if $u$ has one child in $T'$, and of zero states if $u$ has no children in $T'$.
    Thus $$\delta'_{\tup p\tup q}[u]\in ((Q^2 \to Q) \cup (Q \to Q) \cup Q).$$
    Note that for fixed $\tup q$ and $u$, $\delta'_{\tup p\tup q}[u]$ is uniquely determined by the subset of variables of $\tup x$ that $\tup p$ maps to~$u$. This is because $\tup p\in A^{\tup x}$, while
    $u$ is the only node of $\phi^{-1}(u)$ that may belong to $A$.
    Hence, with $u$ we can associate a function $f_u\in \Delta$ that given $\tup t\in \set{0,1}^{\xb}$,
    outputs the transformation $\delta'_{\tup p\tup q}[u]$ for any (equivalently, every) $\tup p\in A^{\tup x}$ satisfying $\tup t(x)=1$ iff $\tup p(x)=u$, for all $x\in \tup x$.
    Then we define the $\Delta$-labeled tree $T'_{\tup q}$ as $T'$ with labeling $u\mapsto f_u$.
    Note that the above construction can be applied to $\tup q=\bl$ in the same way.
    
    Now, for $\tup p\in A^{\tup x}\cup \{\bl\}$ we define the $\Delta\times \set{0,1}^{\tup x}$-labeled tree $(T'_{\tup q})_{\tup p}$ by augmenting $T'_{\tup q}$ with markers for the valuation $\tup p$;
    note that this is possible because $A$ is contained in the node set of $T'$. We also define an automaton $\Ac'$ working on $\Delta\times \set{0,1}^{\tup x}$-labeled trees as follows.
    $\Ac'$ uses the same state set as $\Ac$, while its transition function is defined by taking the binary valuation for $\xb$ in a given node $u$, 
    applying it to the $\Delta$-label of $u$ to obtain a state transformation, verifying that the arity of this transformation matches the number of children of $u$, 
    and finally applying that transformation to the input states. Then the following claim follows immediately from the construction.
    
    \begin{claim}\label{cl:emulate}
     For all $\tup p\in A^{\tup x}\cup \{\bl\}$ and $\tup q\in B^{\tup y}\cup \{\bl\}$, the run of $\Ac'$ on $(T'_{\tup q})_{\tup p}$ is equal to the restriction of the run of $\Ac$ on $T_{\tup p\tup q}$ to the nodes of $B$.
    \end{claim}

    From Claim~\ref{cl:emulate} it follows that if for two tuples $\tup q,\tup q'$ we have $T'_{\tup q}=T'_{\tup q'}$, 
    then for every $\tup p\in A^{\tup x}$, $\Ac$ accepts $T_{\tup p\tup q}$ if and only if $\Ac$ accepts $T_{\tup p\tup q'}$.
    As $\Ac$ is equivalent to the formula $\varphi(\tup x,\tup y)$ in the sense of \cref{lem:grohe-mso-to-automaton}, this implies that 
    $$\{ \tup p\in A^{\tup x}\colon T\models \varphi(\tup p,\tup q)\} = \{ \tup p\in A^{\tup x}\colon T\models \varphi(\tup p,\tup q')\}.$$
    In other words, $\tup q$ and $\tup q'$ define the same element of $S^{\varphi}(T)[A]$. 
    We conclude that the cardinality of $S^{\varphi}(T)[A]$ is bounded by the number of different trees $T'_{\tup q}$ that one can obtain by choosing different $\tup q\in V(T)^{\tup y}$.
    
    Observe that for each $\tup q\in V(T)^{\tup y}$, tree $T'_{\tup q}$ differs from $T'_{\bl}$ by changing the labels of at most $|\tup y|$ nodes.
    Indeed, from the construction of $T'_{\tup q}$ it follows that for each $u\in B$, the labels of $u$ in $T'_{\tup q}$ and in $T'_{\bl}$ 
    may differ only if $\tup q$ maps some variable of $\tup y$ to a node belonging to $\phi^{-1}(u)$;
    this can happen for at most $|\tup y|$ nodes of~$B$.
    Recalling that $|B| \leq 2|A|$ and $|\Delta| \leq |Q|^{2^{|\tup x|}\cdot (|Q|^2+|Q|+1)}$, the number of different trees $T'_{\tup q}$ is bounded by
    $$\sum_{i=0}^{|\yb|} {|B| \choose i} \cdot \left(|Q|^{2^{|\tup x|}\cdot (|Q|^2+|Q|+1)}\right)^{|\yb|} \leq c\cdot |A|^{|\yb|},$$
    where $c\coloneqq 2^{|\yb|}\cdot (|\yb|+1)\cdot \left(|Q|^{2^{|\tup x|}\cdot (|Q|^2+|Q|+1)}\right)^{|\yb|}$.
    As argued, this number is also an upper bound on the cardinality of $S^\varphi(T)[A]$, which concludes the proof.
\end{proof}

\subsection{Classes with bounded treewidth or cliquewidth}

We now exploit the known connections between trees and graphs of bounded treewidth or cliquewidth, expressed in terms of the existence of suitable $\MSO$-transductions, 
to lift \cref{thm:vc-density-in-trees} to more general classes of graphs, thereby proving \cref{thm:upper-bounds}.
In fact, we will not rely on the original combinatorial definitions of these parameters, but on their logical characterizations proved in subsequent works.

The first parameter of interest is the {\em{cliquewidth}} of a graph, introduced by Courcelle and Olariu~\cite{CourcelleO00}.
We will use the following well-known logical characterization of cliquewidth.

\begin{theorem}[\cite{CourcelleE95,EngelfrietO97}]
  \label{thm:cw-tree-interpretable}
  For every $k \in \N$ there is a finite alphabet $\Sigma_k$ and a deterministic $\MSO$-transduction $\Isf_k$ such that for every graph $G$ of cliquewidth at most $k$ 
  there exists a $\Sigma_k$-labeled binary tree $T$ satisfying the following: $\Isf_k(T)$ is the adjacency encoding of $G$.
\end{theorem}

Thus, one may think of graphs of bounded cliquewidth as of graphs that are $\MSO$-interpretable in labeled trees. By combining \cref{thm:cw-tree-interpretable} with \cref{thm:vc-density-in-trees}
we can prove part (i) of \cref{thm:upper-bounds} as follows. 

Fix a class $\Cc$ with uniformly bounded cliquewidth and a partitioned $\CMSO$-formula $\varphi(\xb,\yb)$ over the signature of $\Cc$.
Let $k$ be the upper bound on the cliquewidth of graphs from $\Cc$, and let $\Sigma_k$ and $\Isf_k$ be the alphabet and the deterministic $\MSO$-transduction provided by  \cref{thm:cw-tree-interpretable} for $k$.
Then for every $G\in \Cc$, we can find a $\Sigma_k$-labeled tree $T$ such that $\Isf_k(T)$ is the adjacency encoding of $G$. Note that $V(G)\subseteq V(T)$.
Observe that for every and vertex subset $A\subseteq V(G)$, we have
$$S^{\varphi}(G)[A]\ \subseteq\ S^{\Isf_k^{-1}(\varphi)}(T)[A],$$
where $\Isf_k^{-1}(\varphi)$ is the formula $\varphi$ pulled back through the transduction $\Isf_k$, as given by \cref{lem:btl}.
As by \cref{thm:vc-density-in-trees} we have $|S^{\Isf_k^{-1}(\varphi)}(T)[A]|\leq c\cdot |A|^{|\tup y|}$ for some constant $c$, the same upper bound can be also concluded for the cardinality of $S^{\varphi}(G)[A]$.
This proves \cref{thm:upper-bounds}, part (i).

\medskip

To transfer these result to the case of $\CMSO_2$ over graphs of bounded treewidth, we need to define an additional graph transformation.
For a graph $G$, the {\em{incidence graph}} of $G$ is the bipartite graph with $V(G)\cup E(G)$ as the vertex set, where a vertex $u$ is adjacent to an edge $e$ if and only if $u$ is an endpoint of~$e$.
The following result links $\CMSO_2$ on a graph with $\CMSO_1$ on its incidence graph.

\begin{lemma}[\cite{Courcelle18,Courcelle18a}]\label{lem:incidence-reduce}
    Let $G$ be a graph of treewidth $k$. Then the cliquewidth of the incidence graph of $G$ is at most $k + 3$. 
    Moreover, with any $\CMSO_2$-formula $\varphi(\xb)$ one can associate a $\CMSO_1$-formula $\psi(\xb)$ such that 
    for any graph $H$ and $\bar{a} \in V(H)^{\xb}$ we have $H \models \varphi(\bar{a})$ if and only if $H' \models \psi(\bar{a})$, where $H'$ is the incidence graph of $H$.
\end{lemma}

Now \cref{lem:incidence-reduce} immediately reduces part (ii) of \cref{thm:upper-bounds} to part (i).
Indeed, for every partitioned $\CMSO_2$-formula $\varphi(\tup x,\tup y)$, the corresponding $\CMSO_1$-formula $\psi(\tup x,\tup y)$ provided by \cref{lem:incidence-reduce} satisfies the following: for every graph
$H$ and its incidence graph $H'$, we have
$$S^{\varphi}(H)\subseteq S^{\psi}(H').$$
Observe that by \cref{lem:incidence-reduce}, if a graph class $\Cc$ has uniformly bounded treewidth, then the class $\Cc'$ comprising the incidence graphs of graphs from $\Cc$ has uniformly bounded cliquewidth.
Hence we can apply part (i) of \cref{thm:upper-bounds} to the class $\Cc'$ and obtain an upper bound of the form $|S^{\psi}(H')[A]|\leq c\cdot |A|^{|\tup y|}$ for any $A\subseteq V(H')$, where $c$ is a constant.
By the above containment of set systems, this upper bound carries over to restrictions of~$S^{\varphi}(H)$. This concludes the proof of part (ii) of \cref{thm:upper-bounds}.

%% file: lower.tex
\section{Lower bounds}

We now turn to proving \cref{thm:lower-bounds}. 
As in the work of Grohe and Tur\'an~\cite{Grohe04-mso-tree-vc-dim}, the main idea is to show that the structures responsible for unbounded VC dimension of $\MSO$-definable set systems are {\em{grids}}.
That is, the first step is to prove a suitable unboundedness result for the class of grids, which was done explicitly by Grohe and Tur\'an in~\cite[Example~19]{Grohe04-mso-tree-vc-dim}.
Second, if the considered graph class $\Cc$ has unbounded treewidth (resp., cliquewidth), then we give a deterministic $\MSO_2$-transduction (resp. $\CtwoMSO_1$-transduction) from $\Cc$ to the class of grids.
Such transductions are present in the literature and follow from known forbidden-structures theorems for treewidth and cliquewidth. 
Then we can combine these two steps into the proof of \cref{thm:lower-bounds} using the following generic statement.
In the following, we shall say that logic $\Lc$ has {\em{unbounded VC dimension}} on a class of structures $\Cc$ if there exists a partitioned $\Lc$-formula $\varphi(\tup x,\tup y)$ over the signature of $\Cc$ such that 
the class of set systems $S^\varphi(\Cc)$ has infinite VC dimension.

\begin{lemma}\label{lem:pull-back-lb}
 Let $\Cc$ and $\Dd$ be two classes of structures and $\Lc\in \{\MSO,\CMSO,\CtwoMSO\}$. 
 Suppose that there exists a deterministic $\Lc$-transduction $\Isf$ with input signature being the signature of $\Cc$ and the output signature being the signature of $\Dd$
 such that $\Isf(\Cc)\supseteq \Dd$. Then if $\Lc$ has unbounded VC dimension on $\Dd$, then $\Lc$ also has unbounded VC dimension on $\Cc$.
\end{lemma}
\begin{proof}
 Let formula $\psi(\tup x,\tup y)$ witness that $\Lc$ has unbounded VC dimension on $\Dd$.
 Then it is easy to see that the formula $\varphi\coloneqq \Isf^{-1}(\psi)$, provided by \cref{lem:btl}, witnesses that $\Lc$ has unbounded VC dimension on $\Cc$.
\end{proof}

\subsection{Grids}

\newcommand{\Hor}{\mathsf{H}}
\newcommand{\Ver}{\mathsf{V}}

For $n\in \N$, we denote $[n]\coloneqq \{1,\ldots,n\}$.
An $n\times n$ {\em{grid}} is a relational structure over the universe $[n]\times [n]$ with two successor relations.
The horizontal successor relation $\Hor(\cdot,\cdot)$ selects all pairs of elements of the form $(i,j),(i+1,j)$, where $i\in [n-1]$ and $j\in [n]$.
Similarly, the vertical successor relation $\Ver(\cdot,\cdot)$ selects all pairs of elements the form $(i,j),(i,j+1)$, where $i\in [n]$ and $j\in [n-1]$.
Note that these relations are not symmetric: the second element in the pair must be the successor of the first in the given direction.

Grohe and Tur\'an proved the following.

\begin{theorem}[Example~19 in \cite{Grohe04-mso-tree-vc-dim}]\label{thm:grids-hard}
 $\MSO$ has unbounded VC dimension on the class of grids.
\end{theorem}

The proof of \cref{thm:grids-hard} roughly goes as follows. 
The key idea is that for a given set of elements $X$ it is easy to verify in $\MSO$ the following property: $(i,j)\in X$ is true if and only if the $i$th bit of the binary encoding of $j$ is $1$.
This can be done on the row-by-row basis, by expressing that elements of $X$ in every row encode, in binary, a number that is one larger than what the elements of $X$ encoded in the previous row.
Using this observation, one can easily write a formula $\varphi(x,y)$ that selects exactly pairs of the form $((i,0),(0,j))$ such that $(i,j)\in X$. Then 
$\varphi(x,y)$ shatters the set $\{(i,0)\colon 1\leq i\leq \lfloor \log n\rfloor\}$, as the binary encodings of numbers from $1$ to $n$ give all possible bit vectors of length $\lfloor \log n\rfloor$ when restricted 
to the first $\lfloor \log n\rfloor$ bits. Consequently, $\varphi(x,y)$ shatters a set of size $\lfloor \log n\rfloor$ in an $n\times n$ grid, which enables us to deduce the following slight strengthening of \cref{thm:grids-hard}:
$\MSO$ has unbounded VC dimension on any class of structures that contains infinitely many different grids.

For the purpose of using existing results from the literature, it will be convenient to work with {\em{grid graphs}} instead of grids. 
An $n\times n$ {\em{grid graph}} is a graph on vertex set $[n]\times [n]$ where two vertices $(i,j)$ and $(i',j')$ are adjacent if and only if $|i-i'|+|j-j'|=1$.
When speaking about grid graphs, we assume the adjacency encoding as relational structures. Thus, the difference between grid graphs and grids is that the former are only equipped with a symmetric adjacency relation 
without distinguishement of directions, while in the latter we may use (oriented) successor relations, different for both directions.
Fortunately, grid graphs can be reduced to grids using a well-known construction, as explained next.

\begin{lemma}\label{lem:grid-graphs-to-grids}
 There exists a non-deterministic $\MSO$ transduction $\Jsf$ from the adjacency encodings of graphs to grids such that for every class of graphs $\Cc$ that contains arbitrarily large grid graphs,
 the class $\Jsf(\Cc)$ contains arbitrarily large grids.
\end{lemma}
\begin{proof}
 The transduction uses six additional unary predicates, that is, $\Gamma(\Jsf)=\{A_0,A_1,A_2,B_0,B_1,B_2\}$. We explain how the transduction works on grid graphs, which gives rise to a formal definition of the transduction
 in a straightforward way.
 
 Given an $n\times n$ grid graph $G$, the transduction non-deterministically chooses the valuation of the predicates of $\Gamma(\Jsf)$ as follows: 
 for $t\in \{0,1,2\}$, $A_t$ selects all vertices $(i,j)$ such that $i\equiv t\bmod 3$ and $B_t$ selects all vertices $(i,j)$ such that $j\equiv t\bmod 3$.
 Then the horizontal successor relation $\Hor(\cdot,\cdot)$ can be interpreted as follows: $\Hor(u,v)$ holds if and only if $u$ and $v$ are adjacent in $G$, 
 $u$ and $v$ are both selected by $B_s$ for some $s\in \{0,1,2\}$, and there is $t\in \{0,1,2\}$ such that $u$ is selected by $A_t$ while $v$ is selected by $A_{t+1\bmod 3}$.
 The vertical successor relation is interpreted analogously.
 
 It is easy to see that if $G$ is an $n\times n$ grid graph and the valuation of the predicates of $\Gamma(\Jsf)$ is selected as above, then $\Jsf$ indeed outputs an $n\times n$ grid.
 This implies that if $\Cc$ contains infinitely many different grid graphs, then $\Jsf(\Cc)$ contains infinitely many different grids.
\end{proof}

We may now combine \cref{lem:grid-graphs-to-grids} with \cref{thm:grids-hard} to show the following.

\begin{lemma}\label{lem:grids-hard}
 Suppose $\Lc\in \{\MSO,\CtwoMSO,\CMSO\}$ and $\Cc$ is a class of structures such that there exists a non-deterministic $\Lc$-transduction $\Isf$ from $\Cc$ to adjacency encodings of graphs 
 such that $\Isf(\Cc)$ contains infinitely many different grid graphs. Then there exists a finite signature $\Gamma$ consisting only of unary relation names such that $\Lc$ has unbounded VC dimension on $\Cc^{\Gamma}$.
\end{lemma}
\begin{proof}
 As non-deterministic transductions are closed under composition for all the three considered variants of logic (see e.g. \cite{Courcelle94-mso-transduction-survey}),
 from \cref{lem:grid-graphs-to-grids} we infer that there exists a non-deterministic $\Lc$-transduction $\Ksf$ such that $\Ksf(\Cc)$ contains infinitely many different grids.
 By definition, transduction $\Ksf$ has its deterministic part $\Ksf'$ such that $\Ksf(\Cc)=\Ksf'(\Cc^{\Gamma(\Ksf)})$.
 It now remains to take $\Gamma\coloneqq \Gamma(\Ksf)$ and use \cref{lem:pull-back-lb} together with \cref{thm:grids-hard} (and the remark after it).
\end{proof}

\subsection{Classes with unbounded treewidth and cliquewidth}

For part (ii) of \cref{thm:lower-bounds} we will use the following standard proposition, which essentially dates back to the work of Seese~\cite{Seese91-mso-tw-undecidable}.

\begin{lemma}\label{lem:tw-grid}
 There exists a non-deterministic $\MSO$-transduction $\Isf$ from incidence encodings of graphs to adjacency encodings of graphs such that
 for every graph class $\Cc$ whose treewidth is not uniformly bounded, the class $\Isf(\Cc)$ contains all grid graphs.
\end{lemma}
\begin{proof}
 Recall that a {\em{minor model}} of a graph $H$ in a graph $G$ is a mapping $\phi$ from $V(H)$ to connected subgraphs of $G$ such that subgraphs $\{\phi(u)\colon u\in V(H)\}$ are pairwise disjoint,
 and for every edge $uv\in E(H)$ there is an edge in $G$ with one endpoint in $\phi(u)$ and the other in $\phi(v)$. Then $G$ contains $H$ as a {\em{minor}} if there is a minor model of $H$ in $G$.
 By the Excluded Grid Minor Theorem~\cite{Robertson86-excluded-grid-theorem}, if a class of graphs $\Cc$ has unbounded treewidth, then every grid graph is a minor of some graph from $\Cc$.
 Therefore, it suffices to give a non-deterministic $\MSO$-transduction $\Isf$ from incidence encodings of graphs to adjacency encodings of graphs such that for every graph $G$,
 $\Isf(G)$ contains all minors of $G$.
 
 The transduction $\Isf$ works as follows. Suppose $G$ is a given graph and $\phi$ is a minor model of some graph $H$ in $G$.
 First, in $G$ we non-deterministically guess three subsets: 
 \begin{itemize}[nosep]
  \item a subset $D$ of vertices, containing one arbitrary vertex from each subgraph of $\{\phi(u)\colon u\in V(H)\}$;
  \item a subset $F$ of edges, consisting of the union of spanning trees of subgraphs $\{\phi(u)\colon u\in V(H)\}$ (where each spanning tree is chosen arbitrarily);
  \item a subset $L$ of edges, consisting of one edge connecting a vertex of $\phi(u)$ and a vertex of $\phi(v)$ for each edge $uv\in E(H)$, chosen arbitrarily.
 \end{itemize}
 Recall that graph $G$ is given by its incidence encoding, hence these subsets can be guessed using three unary predicates in $\Gamma(\Isf)$.
 Now with sets $D,F,L$ in place, the adjacency encoding of the minor $H$ can be interpreted as follows: the vertex set of $H$ is $D$, while two vertices $u,u'\in D$ are adjacent in $H$ if and only if 
 in $G$ they can be connected by a path that traverses only edges of $F$ and one edge of $L$. It is straightforward to express this condition in $\MSO_2$.
\end{proof}

Observe that part (ii) of \cref{thm:lower-bounds} follows immediately by combining \cref{lem:tw-grid} with \cref{lem:grids-hard}.
Indeed, from this combination we obtain a partitioned $\MSO$-formula $\varphi(\tup x,\tup y)$ and a finite signature $\Gamma$ consisting of unary relation names such that
the class of set systems $S^{\varphi}(\Cc^\Gamma)$ has infinite VC dimension. Here, we treat $\Cc$ as the class of incidence encodings of graphs from $\Cc$.
Now if we take the label set $\Lambda$ to be the powerset of $\Gamma$, we can naturally modify $\varphi(\tup x,\tup y)$ to an equivalent formula $\varphi'(\tup x,\tup y)$ working over $\Lambda$-ve-labelled graphs,
where the $\Lambda$-label of every vertex $u$ encodes the subset of predicates of $\Gamma$ that select~$u$. 
Thus $S^{\varphi'}(\Cc^{\Lambda,2})$ has infinite VC dimension, which concludes the proof of part (ii) of \cref{thm:lower-bounds}.

To prove part (i) of \cref{thm:lower-bounds} we apply exactly the same reasoning, but with \cref{lem:tw-grid} replaced with the following result of Courcelle and Oum~\cite{Courcelle07-vertex-minors-seese}.

\begin{lemma}[Corollary~7.5 of \cite{Courcelle07-vertex-minors-seese}]
    \label{lem:c2mso1-grids}
    There exists a $\CtwoMSO$-transduction $\Isf$ from adjacency encodings of graphs to adjacency encodings of graphs such that if $\Cc$ is a class of graphs of unbounded cliquewidth, 
    then $\Isf(\Cc)$ contains arbitrarily large grid graphs.
\end{lemma}